\documentclass[secnum,leqno]{amsart}
\usepackage{amssymb, amsmath}
\usepackage{graphics}
\usepackage[dvips]{graphicx}
\usepackage{eepic}
\usepackage{epic}
\def\cD{\mathcal{D}}
\def\cF{\mathcal{F}}
\def\cH{\mathcal{H}}

\def\cL{\mathcal{L}}

\def\dom{\mathrm{Dom}}

\newcommand{\inner}[2]{\left\langle#1,#2\right\rangle}

\newcommand{\hpart}
{\cH_{{\rm part}}}

\def\norm#1{\left\| {#1} \right\|}

\newcommand{\tensor}{\hspace{-0.15mm}\otimes\hspace{-0.15mm}}
\newcommand{\Tensor}{\bigotimes}

\def\bk{\mathbf{k}}        \def\bp{\mathbf{p}}
        \def\bA{\mathbf{A}}

\def\bsigma{\Vec{\sigma}}
\def\be{\mathbf{e}}      \def\bx{\mathbf{x}}    \def\by{\mathbf{y}}

\def\Vec#1{\boldsymbol{#1}}	

\newcommand{\BR}{\mathbb{R}}
\newcommand{\BC}{\mathbb{C}}

\def\ome{\omega}

\def\ep{\epsilon}

\newtheorem{thm}{Theorem}[section]
\newtheorem{crl}[thm]{Corollary}

\newtheorem{lmm}[thm]{Lemma}

\newcommand{\XX}{\BR^3\!\!\times\!\!\{1,2\}}
\theoremstyle{definition}
\newtheorem{dfn}[thm]{Definition}
\theoremstyle{remark}
\newtheorem{rem}{Remark}
\title[ionization energy of the semi-relativistic Pauli-Fierz model]{On the ionization energy of semi-relativistic Pauli-Fierz model for a single particle}

\author[F. Hiroshima]{Fumio Hiroshima}
\address{Fumio Hiroshima: Faculty of Mathematics, Kyushu University 819-0395, Fukuoka, Japan. 
         }
\email{hiroshima@math.kyushu-u.ac.jp}
\author[I. Sasaki]{Itaru Sasaki}
\address{Itaru Sasaki: International Young Researchers Empowerment Center, Shinshu University,
	 390-8621, Matsumoto, Japan. }
\email{isasaki@shinshu-u.ac.jp}	 

\markleft{Hiroshima and Sasaki}

\date{2010, March 8}

\begin{document}

\begin{abstract}      
A semi-relativistic Pauli-Fierz model is defined by the sum of the free Hamiltonian $H_{\rm f}$ of
a Boson Fock space, an nuclear potential $V$  and
a relativistic kinetic energy:
$$
H=\sqrt{[\bsigma\cdot(\bp+e\bA)]^2+M^2}-M+V+
H_{\rm f}.$$
Let $-e_0<0$ be the ground state energy of a semi-relativistic Schr\"odinger operator $$H_{\rm p}=\sqrt{\bp^2+M^2}-M+V.$$
   It is shown that the ionization energy $E^{\rm ion}$ of $H$ satisfies $$E^{\rm ion}\geq e_0>0$$ for all values of both of the coupling constant $e\in\BR$ and
the rest mass $M\geq 0$. In particular
our result includes the case of $M=0$.
\end{abstract}

\maketitle

\section{Introduction}
In this paper we consider a semi-relativistic Pauli-Fierz model in QED.
Throughout SRPF model is a shorthand for semi-relativistic Pauli-Fierz model.
This model describes a dynamics of a semi-relativistic charged particle moving
in the three-dimensional Euclidean space $\BR^3$
under the influence of a real-valued nuclear potential $V$ and  a quantized electromagnetic field.

The Hilbert space $\cH$ of the total system  is the tensor product Hilbert
space of
$\hpart =L^2(\BR^3)\otimes\BC^2$
 and the Boson Fock space
 $\cH_\mathrm{phot}  $ over $L^2(\XX)$.
Here $\hpart $ describes  the state space of a semi-relativistic charged particle with spin $1/2$ and
$\cH_\mathrm{phot}$ that of photons.
 The total Hamiltonian $H^V$ on $\cH$ is given by a minimal coupling to a  quantized electromagnetic field $\bA$ and is of the form
\begin{align}
  H^V =  \sqrt{[\bsigma\cdot(\bp+e\bA)]^2+M^2}
  -M + V + H_{\rm f},
\end{align}
where
$\bp=-i\nabla=(p_1,p_2,p_3)$ denotes the generalized gradient operator,
$e\in\BR$ is the  coupling constant,
$M\geq 0$ the rest mass, $H_\mathrm{f}$
 the free Hamiltonian of
 $\cH_\mathrm{phot}$,
 and $\bsigma=(\bsigma_1,\bsigma_2,\bsigma_3)$ denote  $2\times2$
Pauli matrices satisfying relations: $$\bsigma_1\bsigma_2=i\bsigma_3,\quad
  \bsigma_2\bsigma_3=i\bsigma_1,\quad
  \bsigma_3\bsigma_1=i\bsigma_2.$$
We set $H^0=H^{V=0}$.
Let $E^V=\inf\mathrm{Spec}(H^V)$ be the lowest energy spectrum of $H^V$.
Intuitively $E^0$ describes the energy of
a particle going away from a nucleus
and positivity $E^0-E^V>0$ suggests that ground states of $H^V$ are stable.
The existence of ground states can be indeed shown
under condition
$E^0-E^V>0$.
 The ionization energy is then
defined by
\begin{align}
   E^\mathrm{ion} = E^0 - E^V.
\end{align}
Suppose that
 a semi-relativistic Schr\"odinger  $H_\mathrm{p}=\sqrt{\bp^2+M^2}-M+V$
 has a negative energy ground state $\phi_0$ such that
 $H_\mathrm{p} \phi_0 = -e_0 \phi_0$ with
 $e_0>0$.
Then the main result on this paper is to show that \begin{align}
   E^\mathrm{ion}  \geq e_0 >0.  \label{pos}
\end{align}
It is emphasized that
(\ref{pos})
    is shown for all $(e, M)\in \BR\times[0,\infty)$ and that
    $E^{\rm ion}$ is compared with the lowest energy of
    $H_{\rm p}$.
In order to prove
(\ref{pos}),
 we use a natural extension  of the strategy
developed in \cite{gll} to a semi-relativistic case.
Namely we show that
the ionization energy is greater than
the absolute value of the lowest energy of
the {semi-relativistic} Schr\"odinger operator $H_\mathrm{p}$.

In the case of the non-relativistic Pauli-Fierz model given by
$$\frac{1}{2M}[\bsigma\cdot(\bp+e\bA)]^2+V+H_{\rm f},$$ the positivity of the ionization energy is shown in
\cite{gll,ll}.
We also refer to see the book \cite{ls} and  references therein
for related results.
In the quantum field theory one important
task is to show the existence of ground states,
which
is shown in general
by showing or assuming the positivity of
an  ionization energy. See e.g.,\cite{gll,hs,ll,sasa}.
In this paper we are not concerned with
the existence of ground states of SRPF model,
but this will be done in another paper \cite{sasa2}.

{\it Note added in proof}:
In \cite{kms},  M. K\"onenberg, O. Matte, and E. Stockmeyer also recently prove
$E^\mathrm{ion}>0$
of SRPF model in the case of $V(\bx)=-\gamma/|\bx|$, $\gamma>0$, and $M=1(>0)$, but the ionization energy is not
compared with the lowest energy of
$H_{\rm p}$
but with
     a standard Schr\"odinger operator $\frac{1}{2M}\bp^2+V$.
So it is quite different form ours.
The existence of ground states of SRPF model
is also shown in \cite{kms}.

\section{Definition and Main Result}
We begin with defining SRPF model in a rigorous manner.
We define the Hamiltonian of SRPF model by a quadratic form.

{\bf (Photons)}
The Hilbert space for photons
  is given by \begin{align}
  \cH_\mathrm{phot}
   =  \bigoplus_{n=0}^\infty \left[   \Tensor_\mathrm{s}^n L^2(\XX)\right],
\end{align}
 where  $\tensor_\mathrm{s}^n$ denotes the $n$-fold symmetric tensor product with
$\tensor_\mathrm{s}^0 L^2(\XX)=\BC$.
The smeared annihilation operators in $\cH_\mathrm{phot}$ are denoted by $a(f),  f\in L^2(\XX)$.
The adjoint of $a(f)$, $a^\ast(f)$, is called the creation operator.
The annihilation operator and the creation operator satisfy canonical commutation relations.
The Fock vacuum is defined by
$\Omega_\mathrm{phot}=
1\oplus 0\oplus 0\cdots\in \cH_\mathrm{phot}$.
For a closed operator $T$ on $L^2(\XX)$, the second quantization of $T$ is denoted by \begin{align}
d\Gamma(T):  \cH_\mathrm{phot}\to
  \cH_\mathrm{phot}.
  \end{align}
  Let $\omega:\BR^3\to [0,\infty)$ be a Borel measurable function such that $0<\omega(\bk)<\infty$ for almost every
$\bk\in\BR^3$. We also denote by the same symbol $\omega$ the multiplication operator by the function $\omega$,
which acts  in $L^2(\XX)$ as $(\omega f)(\bk,j)=\omega(\bk)f(\bk,j)$.
Then the free Hamiltonian of $\cH_\mathrm{phot}  $
is defined by
\begin{align}
  H_\mathrm{f} = d\Gamma(\omega).
\end{align}

{\bf (Charged particle)}
The Hilbert space for the particle state
is defined by
\begin{align}
  \hpart  = L^2(\BR^3)\tensor \BC^2, \end{align}
where $\BC^2$ describes spin $1/2$ of the particle.
Then the particle Hamiltonian is given by
a semi-relativistic Schr\"odinger operator: \begin{align}
 H_\mathrm{p} = \sqrt{\bp^2+M^2}-M + V.
\end{align}

{\bf (SRPF model)}
The Hilbert space of the SRPF model is defined by
\begin{align}
  \cH= \hpart  \tensor \cH_\mathrm{phot},
\end{align}
and the decoupled Hamiltonian of the system  is given by
\begin{align}
H_{\rm p}\otimes 1+1\otimes H_{\rm f}.
\end{align}
We introduce an interaction minimally coupled to
$H_{\rm p}\otimes 1+1\otimes H_{\rm f}$.
Let $\be^{(\lambda)}:\BR^3\to\BR^3,\lambda=1,2$,
be polarization vectors such that
\begin{align}
 \be^{(\lambda)}(\bk) \cdot \be^{(\mu)}(\bk) = \delta_{\lambda,\mu}, \quad
 \bk \cdot \be^{(\lambda)}(\bk)=0, \quad
  \lambda,\mu\in \{1,2\}.
\end{align}
We set
 $e^{(\lambda)}(\bk)=(e_1^{(\lambda)}(\bk),e_2^{(\lambda)}(\bk),e_3^{(\lambda)}(\bk))$
 and
suppose that each component $e_j^{(\lambda)}(\bk)$ is a Borel measurable function in $\bk$.
Let $\Lambda\in L^2(\BR^3)$ be a function such that
\begin{align}
 \ome^{-1/2}\Lambda\in L^2(\BR^3). 
\end{align}
We set
\begin{align}
  g_j(\bk,\lambda;\bx) = \omega(\bk)^{-1/2}\Lambda(\bk)e_j^{(\lambda)}(\bk)e^{-i\bk\cdot\bx}.
\end{align}
For each $\bx\in\BR^3$, $g_j(\bx)=g_j(\cdot,\cdot;\bx)$ can be  regarded as an element of
$L^2(\XX)$.
The quantized electromagnetic field at $\bx\in\BR^3$ is defined by
\begin{align}
 \bA_j(\bx)= \frac{1}{\sqrt{2}}\overline{\left[a(g_j(\bx))+a^*(g_j(\bx))\right]},
\end{align}
where  $\bar{T}$ denotes
the closure of closable operator $T$.
The quantized electromagnetic field $\bA(\bx)=(\bA_1(\bx),\bA_2(\bx),\bA_3(\bx))$ satisfies the Coulomb gage condition:
$$\sum_{j=1}^3\frac{\partial
\bA_j(\bx)}{\partial \bx_j}=0.$$
 It is known that $\bA_j(\bx)$ is self-adjoint
for each $\bx\in\BR^3$.
The Hilbert space $\cH$ can be identified as
\begin{align}
  \cH
\cong \int_{\BR^3}^\oplus \BC^2\tensor \cH_\mathrm{phot} d\bx,   \label{ident}
\end{align}
where $\int^\oplus \cdots$ denotes
a constant fiber direct integral \cite{rs4}.
The quantized electromagnetic field on the total Hilbert space is defined by
the constant fiber direct integral of $\bA_j(\bx)$:
\begin{align}
  \bA_j = \int_{\BR^3}^\oplus \bA_j(\bx)d\bx.
\end{align}
Let $C_0^\infty(\BR^3)$ be a set of
infinite times differentiable functions with
a compact support.
The finite particle subspace over  $C_0^\infty=C_0^\infty(\BR^3)$ is defined by
\begin{align}
 \cF_{\rm fin}=\cL[\{ a^*(f_1) \cdots a^*(f_n) \Omega_\mathrm{phot}, \Omega_\mathrm{phot}|
 f_j\in C_0^\infty, j=1,\dots,n, n\in\mathbb N \}],
\end{align}
where $\cL[\cdots]$ denotes the linear hull of
$[\cdots]$. We set
\begin{align}
 \cD = C_0^\infty(\BR^3;\BC^2) \hat\tensor  \cF_{\rm fin},
\end{align}
where the symbol $\hat\tensor$ denotes the algebraic tensor product.
In what follows for notational  convenience
we omit $\otimes$ between $\hpart $ and $ \cH_\mathrm{phot}$, i.e., we write
$p_j$ for $p_j\tensor 1$, $H_\mathrm{f}$ for $1\tensor H_\mathrm{f}$ and $V$ for
$V\tensor 1$, etc.
We define the  non-negative quadratic form on $\cD\times\cD$ by
\begin{align}
 K_{\bA}(\Psi,\Phi) =
 \sum_{j=1}^3
 \inner{
 \bsigma_j(p_j+e\bA_j)\Psi}
 {\bsigma_j(p_j+e\bA_j)
 \Phi} + M^2\inner{\Psi}{\Phi}.
 \label{qua}
 \end{align}
Since $K_{\bA}$ is closable, we denote its closure by $\bar K_{\bA}$.
The semi-bounded quadratic form $\bar K_{\bA}$ defines the unique non-negative self-adjoint operator $H_{\bA}$ such that
$$\dom(H_{\bA}^{1/2})=Q(\bar K_{\bA})$$ and
$$\bar K_{\bA}(\Psi,\Phi)=(H_{\bA}^{1/2} \Psi, H_{\bA}^{1/2}\Phi)$$ for
$\Psi,\Phi\in Q(\bar{K}_{\bA})$.
Here $Q(X)$ denotes the form domain of $X$.
We note that
\begin{eqnarray*}
&&
 Q(H_{\bA}) = \left\{\Psi\in \cH| \exists \{\Psi_n\}_n \subset \cD   \text{ s.t. } \Psi_n \to \Psi,\right.\\
 &&\hspace{5cm}\left.
          K_{\bA}(\Psi_m-\Psi_n,\Psi_m-\Psi_n)\to 0, n,m\to\infty\right\},   \notag \\
&&
 \inner{H_{\bA}^{1/2}\Psi}{H_{\bA}^{1/2}\Psi}= \lim_{n\to\infty}K_{\bA}(\Psi_n,\Psi_n)
 \text{ for }  \Psi\in Q(H_{\bA}).
\end{eqnarray*}
    Since $\cD\subset Q(H_{\bA})$,
     we have
     $ \cD \subset  \dom(H_{\bA}^{1/2})$.
\begin{rem}
It is not trivial to see the essential self-adjointness or the self-adjointness of $[\bsigma\cdot
(\bp+e \bA)]^2+M^2$.
Note that however on $\cD$
we see that
$$H_{\bA}=[\bsigma\cdot (\bp+e \bA)]^2+M^2$$
and then
$H_{\bA}^{1/2}$ can be regarded as
a rigorous definition of
$\sqrt{[\bsigma\cdot
(\bp+e \bA)]^2+M^2}$.
\end{rem}
Now we define the Hamiltonian of SRPF model.
\begin{dfn}{\bf (SRPF model)}\\
(1) Let us define $H^0$ by
\begin{align}
 H^0 =  H_{\bA}^{1/2} -M +H_\mathrm{f}.
\end{align}
(2)
Let $V:\BR^3\to\BR$ be
such that $
   V\in L_\mathrm{loc}^2(\BR^3)$.
Then the Hamiltonian of SRPF model is defined by
\begin{align}
 H^V = H^0 + V.
\end{align}
\end{dfn}

We introduce the following conditions:
\begin{description}
 \item[(H.1)] $H_\mathrm{p}$ is essentially self-adjoint on $C_0^\infty(\BR^3)$.
 \item[(H.2)]
    $H_\mathrm{p}$ has a normalized  negative energy ground state $\phi_0$:
\begin{align}
     H_\mathrm{p} \phi_0 = - e_0\phi_0, \quad e_0>0,\quad
     -e_0=\inf\mathrm{Spec}(H_\mathrm{p}).
\end{align}
\item[(H.3)]
   $H^0$ and $H^V$ are essentially self-adjoint on $\cD$.
We denote the closure of $H^V\lceil_{\cD}$ by the same symbol.
\end{description}
\begin{rem}
Although it is interested in
  specifying conditions for $H^V$
  to be self-adjoint
  or essential self-adjoint    on some domain,
in this paper we do not discuss it.
\end{rem}
Let
$  E^0 = \inf \mathrm{Spec}(H^0)$
and
$ E^V = \inf \mathrm{Spec}(H^V)$.
The ionization energy
is defined by
\begin{align}
 E^\mathrm{ion} = E^0-E^V.
\end{align}
If $V\leq0$, then  it is trivial to see that
$E^{\rm ion}\geq 0$.
The main result in this paper is however
as follows:
\begin{thm}\label{t1}
 Assume (H.1)--(H.3). Then
$
  E^\mathrm{ion} \geq  e_0>0$
for all $(e,M)\in\BR\times[0,\infty)$.
\end{thm}

\section{Proof of Theorem \ref{t1}}
Throughout this section we assume
(H.1)--(H.3).
We fix an arbitrary small $\ep>0$.
Let $F_0$ and $f_0$ be $\ep$-minimizers of $H^0$ and $H_\mathrm{p}$, respectively, i.e.,
\begin{align}
 & \inner{F_0}{H^0 F_0}_{\cH} < E^0+\ep, \quad \norm{F_0}_{\cH}=1, \quad F_0 \in\cD,    \label{min1}\\
 & \inner{f_0}{H_\mathrm{p} f_0}_{\hpart }
   < -e_0 +\ep,     \quad \norm{f_0}_{\hpart }=1, \quad f_0 \in C_0^\infty(\BR^3). \label{min2}
\end{align}
Since $\cD$ and $C_0^\infty(\BR^3)$ are cores for $H^0$ and $H_\mathrm{p}$, respectively,
we can choose a  minimizer satisfying \eqref{min1} and \eqref{min2}.
Recall that  $H_\mathrm{p}$ has a ground state $\phi_0$.
Clearly,
its complex conjugate $\phi_0^*$ is also the ground state of $H_\mathrm{p}$.
Hence
we may assume that $\phi_0$ is real without loss of generality.
Therefore we can choose a real-valued  $\ep$-minimizer $f_0$.
For each $\by\in\BR^3$, we set
\begin{align*}
 U_\by = \exp(-i\by\cdot\bp)\tensor \exp(-i\by\cdot d\Gamma(\bk)).
\end{align*}
The unitary operator $U_\by$ is the parallel translation by the vector $y\in\BR^3$.
It can be shown that $ U_\by \cD =\cD $ and $H^0$ is translation invariant:
\begin{align}
U_\by^* H^0 U_\by =H^0.   \label{trans}
\end{align}
Set
$$
\Omega_\bA(\bp)=
H_{\bA}^{1/2}$$
\begin{lmm}
\label{l0}
Let
$
 \Phi_\by = f_0(\hat{\bx})F_\by
$, where $f_0(\hat{\bx})$ denotes the multiplication by the function $f_0(\bx)$, and $F_\by=U_\by F_0$.
Then we have
\begin{align*}
& \int_{\BR^3}d\by  \inner{\Phi_\by}{H^V\Phi_\by}                                                 \notag \\
&=
  \norm{f_0}^2 \inner{F_0}{H^0F_0}        +   \inner{f_0}{Vf_0}\inner{F_0}{F_0} \\
&\quad 
 +\frac{1}{2} \int_{\BR^3} d\bk  
     |\hat{f}_0(\bk)|^2 \inner{F_0}{[\Omega_\bA(\bp+\bk)+\Omega_\bA(\bp-\bk)-2\Omega_\bA(\bp)]F_0}.         
\end{align*}
\end{lmm}
\begin{proof}
Clearly, $\Phi_\by \in \cD$.
By using \eqref{trans}, we have
\begin{eqnarray}
 \inner{\Phi_\by}{H^V\Phi_\by}
& =& \inner{f_0(\hat{\bx})F_\by}{f_0(\hat{\bx}) H^V F_\by}
   -
   \frac{1}{2}\inner{F_\by}{[f_0(\hat{\bx}) ,[f_0(\hat{\bx}),H^V]] F_\by} \nonumber \\
& = &\inner{f_0(\hat{\bx}+\by)^2 F_0}{ H^0 F_0}
    +  \inner{f_0(\hat{\bx}+\by)^2F_0}{V(\bx+\by)  F_0}\nonumber\\
&   &\quad\quad  -
   \frac{1}{2}\inner{F_0}{[f_0(\hat{\bx}+\by) ,[f_0(\hat{\bx}+\by),\Omega_\bA(\bp)]] F_0}.
\label{i}
\end{eqnarray}
It should be noted that $\cD$ is invariant by the unitary operator $e^{i\bk\cdot\bx}$ and
\begin{eqnarray*}
 \inner{
 \Omega_{\bA}(\bp)
e^{i\bk\cdot\bx}\Psi}{
\Omega_{\bA}(\bp)e^{i\bk\cdot\bx}\Phi}
&=&
\sum_{j=1}^3
\inner{\bsigma_j
(p_j+k_j+e\bA_j)\Psi}
{\bsigma_j(p_j+k_j+e\bA_j)\Phi}\\
&=&
\inner{\Omega_{\bA}(\bp+\bk)\Psi}
{\Omega_{\bA}(\bp+\bk)\Phi}
\end{eqnarray*}
for $\Psi \in \cD$.
Hence by the definition of $\Omega_{\bA}(\bp)$
we have
\begin{align}
\label{good}
  \Omega_\bA(\bp+\bk) = e^{-i\bk\cdot\bx}\Omega_\bA(\bp)e^{i\bk\cdot\bx}.
\end{align}
Thus the last term ({\ref i}) can be computed.
We have by the inverse Fourier transformation, \begin{eqnarray*}
&&
-
   \frac{1}{2}\inner{F_0}{[f_0(\hat{\bx}+\by) ,[f_0(\hat{\bx}+\by),\Omega_\bA(\bp)]] F_0}\\
&&=
  -\frac{1}{2(2\pi)^3} \int_{\BR^6} d\bk_1d\bk_2
    \hat{f}_0(\bk_1)\hat{f}_0(\bk_2) e^{i\bk_1 \cdot \by}e^{i\bk_2 \cdot \by}
    \inner{F_0}{[e^{i\bk_1\cdot \bx},[e^{i\bk_2\cdot \bx},\Omega_\bA(\bp)]]F_0} \notag.
    \end{eqnarray*}
Using (\ref{good}) twice,  we see that
\begin{eqnarray*}
&&=
  -\frac{1}{2(2\pi)^3} \int_{\BR^6} \!\!\! d\bk_1d\bk_2 \notag \\
 &  &\qquad \times \hat{f}_0(\bk_1)\hat{f}_0(\bk_2) e^{i(\bk_1+\bk_2) \cdot \by}
    \inner{F_0}{[e^{i\bk_1 \cdot \bx},\Omega_\bA(\bp-\bk_2)-\Omega_\bA(\bp)]F_0}   \notag \\
&&=
 -\frac{1}{2(2\pi)^3} \int_{\BR^6} d\bk_1d\bk_2
    \hat{f}_0(\bk_1)\hat{f}_0(\bk_2) e^{i(\bk_1+\bk_2) \cdot \by}               \notag \\
& &\qquad  \times    \inner{F_0}{[\Omega_\bA(\bp-\bk_2-\bk_1)-\Omega_\bA(\bp-\bk_2)-\Omega_\bA(\bp-\bk_1)+\Omega_\bA(\bp)]F_0}.
\end{eqnarray*}
Under identification $\cH\cong\int^\oplus_{\BR^3}\BC^2\otimes\cH_{\rm phot}d\bx$,
$F_0$ can be regarded as a $\BC^2\tensor\cH_\mathrm{phot}$-valued
$L^2$-function.
Then
we have
\begin{eqnarray*}
&& \int_{\BR^3}d\by  \inner{\Phi_\by}{H^V\Phi_\by}                                                 \notag \\
&&= \int_{\BR^3} d\by \int_{\BR^3} d\bx  f_0(\bx+\by)^2
\left(
\inner{F_0(\bx)}{(H^0F_0)(\bx)}+
 V(\bx+\by)
          \inner{F_0(\bx)}{F_0(\bx)}\right)
            \notag \\
&&~~~~  -\frac{1}{2(2\pi)^3}  \int_{\BR^3} d\by \int_{\BR^6} d\bk_1d\bk_2
    \hat{f}_0(\bk_1)\hat{f}_0(\bk_2)
    e^{i(\bk_1+\bk_2) \cdot \by} \notag \\
&&\qquad  \times    \inner{F_0}{[\Omega_\bA(\bp-\bk_2-\bk_1)-\Omega_\bA(\bp-\bk_2)-\Omega_\bA(\bp-\bk_1)+\Omega_\bA(\bp)]F_0},
\end{eqnarray*}
where we used the fact that $\hat{f_0}(-\bk)=\hat{f_0}(\bk)^*$.
Hence we have
\begin{align*}
  &\int_{\BR^3}d\by  \inner{\Phi_\by}{H^V\Phi_\by}    \\
  &=  \norm{f_0}^2 \inner{F_0}{H^0F_0}       
   +   \inner{f_0}{Vf_0}\inner{F_0}{F_0}    \\
&\qquad  +\frac{1}{2} \int_{\BR^3} d\bk
     |\hat{f}_0(\bk)|^2 \inner{F_0}{[\Omega_\bA(\bp+\bk)+\Omega_\bA(\bp-\bk)-2\Omega_\bA(\bp)]F_0}.     
\end{align*}
Then the lemma follows.
\end{proof}
The following inequality is the key to the proof of Theorem \ref{t1}
\begin{lmm}\label{l1}
 For all $M\geq 0$ and $\bk\in\BR^3$, the operator inequality
\begin{align}
\frac{1}{2}\left\{ \Omega_\bA(\bp+\bk)+\Omega_\bA(\bp-\bk)-2\Omega_\bA(\bp) \right\}\leq
  \sqrt{\bk^2+M^2}-M,  \label{key}
\end{align}
holds on $\dom(\Omega_\bA(\bp))$.
\end{lmm}
\begin{proof}
Note that the domains of $\Omega_\bA(\bp+\bk)$,  $\Omega_\bA(\bp-\bk)$ and $\Omega_\bA(\bp)$
are identical.
\eqref{key} is equivalent to
\begin{align}
 & \Omega_\bA(\bp+\bk)+\Omega_\bA(\bp-\bk) \leq 2(\sqrt{\bk^2+M^2}-M + \Omega_\bA(\bp)).  \label{key2}
\end{align}
By the Kato-Rellich Theorem, \eqref{key2} follows from
\begin{align}
\norm{[ \Omega_\bA(\bp+\bk)+\Omega_\bA(\bp-\bk)]\Psi}^2 \leq
  \norm{2[\sqrt{\bk^2+M^2}-M + \Omega_\bA(\bp)]\Psi}^2  \label{ineq1}
\end{align}
for
$\Psi \in \dom(\Omega_\bA(\bp))$.
We have the bound
\begin{align}
 \norm{[ \Omega_\bA(\bp+\bk)+\Omega_\bA(\bp-\bk)]\Psi}^2
 &\leq 2[\norm{ \Omega_\bA(\bp+\bk)\Psi}^2 + \norm{\Omega_\bA(\bp-\bk)^2\Psi}^2] \notag \\
 & = 4\bk^2\norm{\Psi}^2 + 4\norm{\Omega_\bA(\bp)\Psi}^2,\label{ineq2}
\end{align}
for all $\Psi\in\dom(\Omega_\bA(\bp))$.
While we have 
\begin{align}
& 4[\sqrt{\bk^2+M^2}-M + \Omega_\bA(\bp)]^2 \notag \\
& =  4[ \bk^2 + [\bsigma\cdot(\bp+e\bA)]^2 +M^2
       +2(\sqrt{\bk^2+M^2}-M)(\Omega_\bA(\bp)-M)],         \label{ineq3}
\end{align}
in the sense of form on $\dom(\Omega_\bA(\bp))$.
Since $\Omega_\bA(\bp)-M$ is positive, inequality \eqref{ineq2} and equality \eqref{ineq3}
imply \eqref{ineq1}. Therefore inequality \eqref{key} holds.
\end{proof}
\begin{crl}
\label{hayai}
It follows that
\begin{eqnarray*}
   \int_{\BR^3}d\by  \inner{\Phi_\by}{H^V\Phi_\by}
\leq
  \inner{F_0}{H^0F_0}        +   \inner{f_0}{H_{\rm p}f_0}.
\end{eqnarray*}
\end{crl}
\begin{proof}
By using Lemmas \ref{l0} and \ref{l1},
 we have
\begin{align*}
   \int_{\BR^3}d\by  \inner{\Phi_\by}{H^V\Phi_\by}
  \leq&  \inner{F_0}{H^0F_0}        +   \inner{f_0}{Vf_0}
    +\int_{\BR^3} d\bk  (\sqrt{\bk^2+M^2}-M) |\hat{f}_0(\bk)|^2\\
    =&
\inner{F_0}{H^0F_0}        +   \inner{f_0}{H_{\rm p}f_0}.
\end{align*}
Then the corollary follows.
\end{proof}
{\it Proof of Theorem \ref{t1}:}
By Corollary \ref{hayai}
and the definitions of $F_0$ and $f_0$, we have
\begin{align*}
    \int_{\BR^3}d\by \left[-\inner{\Phi_\by}{H^V\Phi_\by} +(E^0  -e_0+2\ep)\norm{\Phi_\by}^2 \right]  >0.
\end{align*}
Therefore, there exists a vector $y\in\BR^3$ such that $\Phi_\by\neq 0$ and
\begin{align*}
  E^V\norm{\Phi_\by}^2 \leq    \inner{\Phi_\by}{H^V\Phi_\by}  <  (E^0  -e_0+2\ep)\norm{\Phi_\by}^2.
\end{align*}
Since $\ep$ is arbitrary, this yields \eqref{key} and completes the proof of the theorem.

\section*{Comments and Acknowledgments}
This paper provides an improved version of the result presented in the international  conference
``Applications of RG Methods in Mathematical Science'' held in Kyoto University in Sept. 2009.
IS is grateful to K. R. Ito for inviting me to the conference.
We are grateful to T. Miyao for bring \cite{kms} to our attention.
This study was performed through Special Coordination Funds for
Promoting Science and Technology of the Ministry of Education,
Culture, Sports, Science and Technology, the Japanese Government.
 FH acknowledges support of Grant-in-Aid for
Science Research (B) 20340032 from JSPS.



\begin{thebibliography}{99}
%
%

\bibitem{gll}
M.  Griesemer,   E. H.  Lieb and M.  Loss,
Ground states in non-relativistic quantum electrodynamics,
{\it Invent. Math.} \textbf{145} (2001), 557--595.

\bibitem{hs} F. Hiroshima and I. Sasaki,
Enhanced binding of an $N$-particle system interacting with a scalar field I,
\textit{Math. Z.} \textbf{259} (2008), 657-680.

\bibitem{sasa2}
F. Hiroshima and  I. Sasaki,
  Ground state of semi-relativistic Pauli-Fierz model,
in preparation.

\bibitem{kms} M. K\"onenberg, O. Matte and E. Stockmeyer,
Existence of ground states of hydrogen-like atom in relativistic QED I:
The semi-relativistic Pauli-Fierz operator,
arXiv:math-ph/09124223.v1, preprint  2009.

\bibitem{ll}
E. H. Lieb and M. Loss,
Existence of atoms and molecules in non-relativistic quantum electrodynamics,
\textit{Adv. Theor. Math. Phys.}
\textbf{7} (2003), 667-710.

\bibitem{ls}
E. H. Lieb and R. Seiringer,
{\it The stability of matter in quantum mechanics},
\textit{Cambridge University Press.} 2010.

\bibitem{mls} M. Loss, T. Miyao and H. Spohn,
Lowest energy states in nonrelativistic QED: Atoms and ions in motion,
\textit{J. Funct. Anal.}
\textbf{243} (2007), 353-393.

\bibitem{rs4}
M. Reed and B. Simon,
{\it Method of Modern Mathematical Physics IV},
Academic press, 1978.

\bibitem{sasa}
I. Sasaki, Ground state of the massless Nelson model in a non-Fock representation,
{\it J. Math. Phys.} \textbf{46} (2005),
102107.

\end{thebibliography}
\end{document}